\documentclass[lettersize,journal]{IEEEtran}
\usepackage[table]{xcolor}  
\usepackage{colortbl}       
\usepackage{amsmath,amsfonts}  
\usepackage{algorithmic}
\usepackage{algorithm}
\usepackage{array}
\usepackage[caption=false,font=normalsize,labelfont=sf,textfont=sf]{subfig}
\usepackage{textcomp}
\usepackage{stfloats}
\usepackage{url}
\usepackage{float}
\usepackage{ulem}
\usepackage{wrapfig}
\usepackage{amsthm}
\usepackage{graphicx}
\usepackage{booktabs}
\definecolor{mygreen}{rgb}{0,0.6,0}
\definecolor{mygray}{rgb}{0.5,0.5,0.5}
\definecolor{mymauve}{rgb}{0.58,0,0.82}
\usepackage{overpic}
\usepackage{caption}
\newtheorem{theorem}{Theorem}

\definecolor{AAA}{rgb}{1.0, 0.13, 0.32}
\definecolor{BBB}{rgb}{0.2, 0.1, 1}
\definecolor{CCC}{rgb}{0.0, 1, 0}
\definecolor{DDD}{rgb}{0.0, 0, 1}

\definecolor{headerblue}{rgb}{0.8,0.87,0.94}
\definecolor{rowgray}{rgb}{0.95,0.95,0.95}

\usepackage{verbatim}
\usepackage{graphicx}
\usepackage{cite}
\begin{document}

\title{Power Diagram Enhanced Adaptive Isosurface Extraction from Signed Distance Fields}


\author{Pengfei Wang, Ziyang Zhang, Wensong Wang, Shuangmin Chen, Lin Lu, Shiqing Xin$^*$, Changhe Tu, Wenping Wang
\IEEEcompsocitemizethanks{
\IEEEcompsocthanksitem P. Wang, Z. Zhang, W. Wang, L. Lu, S. Xin, C. Tu are with the School of Computer Science, Shandong University. S. Xin is the corresponding author (email: xinshiqing@sdu.edu.cn).
\IEEEcompsocthanksitem S. Chen is with the School of Information and Technology, Qingdao University of Science and Technology.
\IEEEcompsocthanksitem W. Wang is with the Computer Science and Engineering, Texas A\&M University.
}
}




\twocolumn[{%
\renewcommand\twocolumn[1][]{#1}%
\maketitle
\includegraphics[width=.98\linewidth]{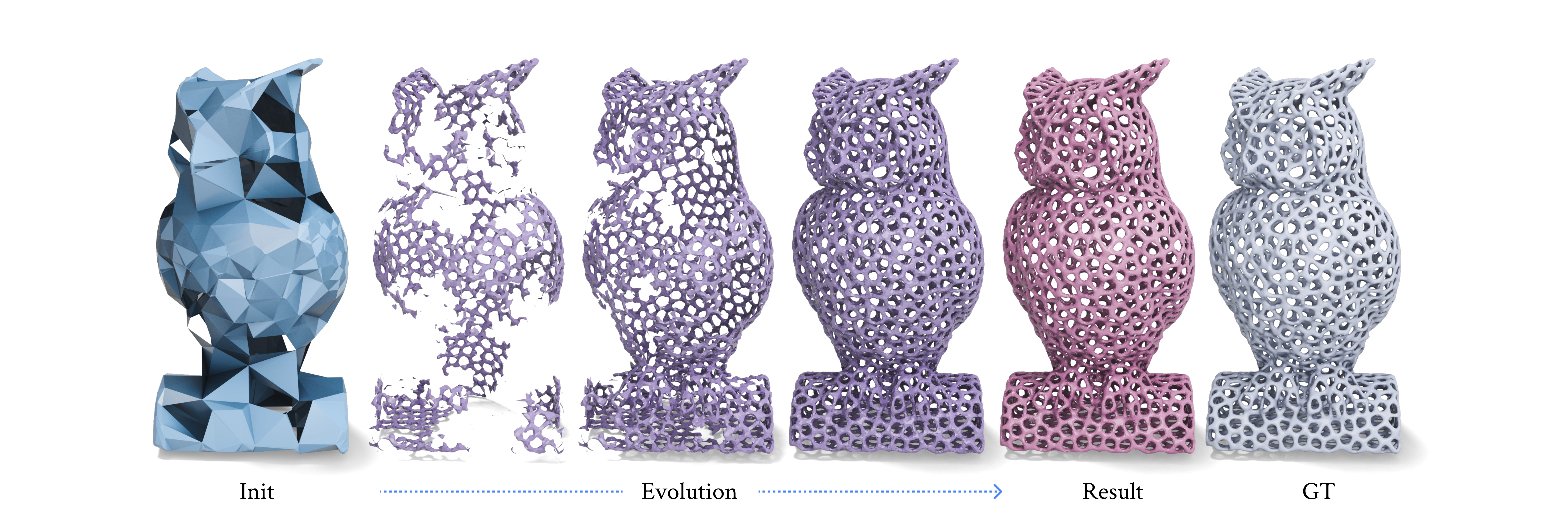}
\captionof{figure}{The workflow of our isosurface extraction. Starting from an initial coarse mesh surface, our method iteratively refines the geometry, progressively converging toward topologically faithful and geometrically accurate results. }
\label{fig:teaser}}]

\begin{abstract}
Extracting high-fidelity mesh surfaces from Signed Distance Fields (SDFs) has become a fundamental operation in geometry processing. Despite significant progress over the past decades, key challenges remain—namely, how to automatically capture the intricate geometric and topological structures encoded in the zero level set of SDFs. In this paper, we present a novel isosurface extraction algorithm that introduces two key innovations:
1) An incrementally constructed power diagram through the addition of sample points, which enables repeated updates to the extracted surface via its dual—regular Delaunay tetrahedralization; and  
2) An adaptive point insertion strategy that identifies regions exhibiting the greatest discrepancy between the current mesh and the underlying continuous surface.
As the teaser figure shows, our framework progressively refines the extracted mesh with minimal computational cost until it sufficiently approximates the underlying surface. Experimental results demonstrate that our approach outperforms state-of-the-art methods, particularly for models with intricate geometric variations and complex topologies.
\end{abstract}


\section{Introduction}
\label{sec:intro}

Signed Distance Fields (SDFs) serve as a fundamental implicit surface representation in geometry processing~\cite{729581,10.1145/3596711.3596732,1634323,LevelSetMethodsandFastMarchingMethods,TheLevelSetMethodsandDynamicImplicitSurfaces,559540,10.1145/3610548.3618170,Park_2019_CVPR}. With the development of geometric deep learning techniques, neural SDFs~\cite{takikawa2021nglod,wang2021neus} have become popular in many tasks, including surface reconstruction, shape completion, and shape generation. Consequently, extracting a topologically faithful and geometrically accurate mesh surface from SDFs has become increasingly important.

However, during mesh extraction, fully preserving geometric fidelity and capturing the true topological structure remains challenging due to the inherent implicit nature of SDFs. Traditional approaches, such as Marching Cubes~\cite{10.1145/37401.37422,DBLP:journals/corr/abs-2106-11272,doi1991efficient} and Dual Contouring~\cite{10.1145/566654.566586}, as well as some learning-based variants~\cite{liao2018deep,chen2022ndc,remelli2020meshsdf}, extract the mesh surface by discretizing space into uniform cubic units and analyzing how the surface intersects the grid cells. They assume that the surface within each grid cell can be deemed sufficiently simple, thereby heavily relying on the pre-defined resolution of the grid. When the geometry is intricate—for instance, when the local feature size (LFS) is smaller than the cell size—geometric features (such as geometric details, feature lines, tip points, tiny tubes, and thin plates) may not be completely recovered or might even be lost in the extracted triangle mesh.

Recently, several geometry-aware mesh extraction algorithms~\cite{sellán2023reachspherestangencyawaresurface, Sellan2024RFTA, renınst2024mcgridsmontecarlodrivenadaptive} have been proposed for adaptive mesh extraction. For instance, Reach for the Spheres (RFTS)~\cite{sellán2023reachspherestangencyawaresurface} operates under the assumption that each SDF sample corresponds to a spherical region that lies entirely either inside or outside the surface. Its extension, Reach for the Arc (RFTA)~\cite{Sellan2024RFTA}, leverages additional information from the SDF samples. Both methods perform well at low resolutions, when only a limited number of SDF samples are queried. However, RFTS struggles to capture high-frequency geometric variations, and RFTA is unable to modify the initial topology during the refinement process. In contrast, McGrids~\cite{renınst2024mcgridsmontecarlodrivenadaptive} formulates adaptive grid construction as a probability sampling problem, solved via a Monte Carlo process. Despite its advantages, it still struggles to resolve tiny details and small topological features.

\begin{figure}[h]
	\centering
\begin{overpic}
[width=.98\linewidth]{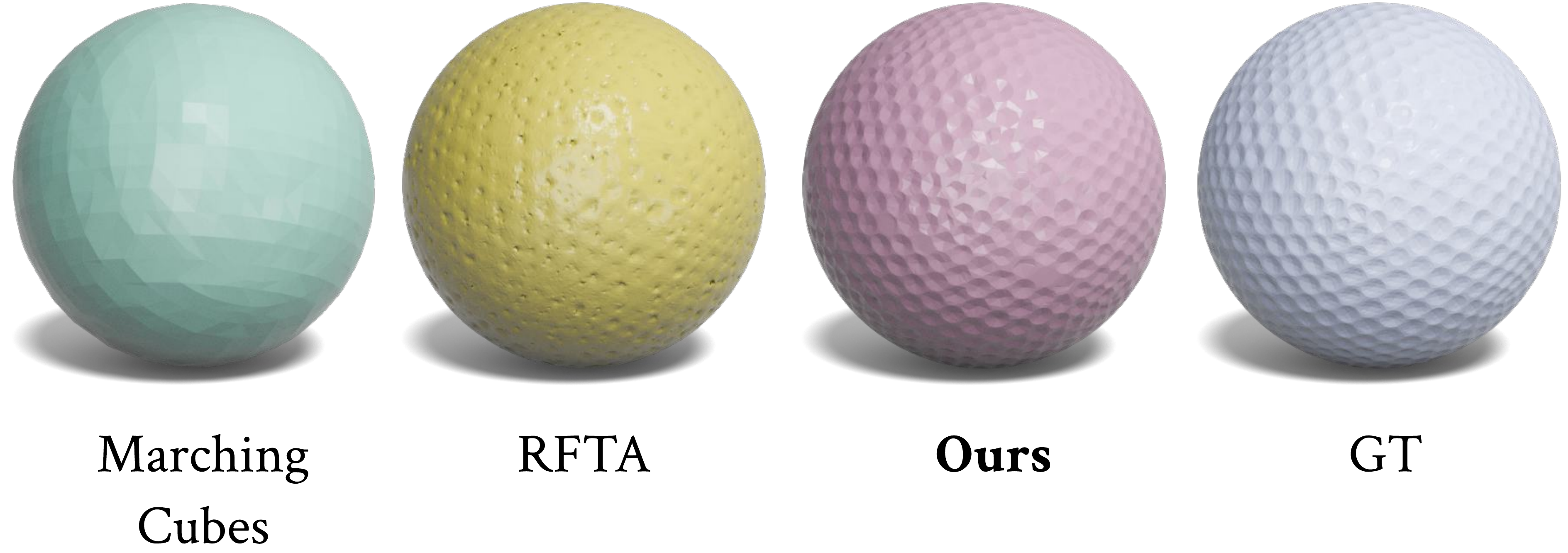}
\end{overpic}
\caption{
Compared with existing isosurface extraction approaches, our method excels at recovering intricate surface details.
}
\vspace{-3mm}
\label{fig:golf}
\end{figure}

In this paper, we improve adaptive mesh extraction approaches through two key innovations. First, we exploit the observation that a collection of samples with signed distance function (SDF) values defines a power diagram where squared SDF values act as weights. This power diagram contains a subset of facets that approximate the underlying surface. By incrementally constructing the power diagram with added samples, its dual tetrahedralization induces an adaptive and refined grid structure, enabling mesh extraction that closely aligns with the true surface. 
Second, we propose a novel point insertion strategy that prioritizes regions with the largest discrepancy between the current mesh and the underlying continuous surface. This simple yet effective strategy optimizes sample placement to maximize reconstruction accuracy.  
We further integrate these techniques into a progressive framework that requires only local updates per iteration, minimizing computational overhead. The algorithm terminates once a predefined accuracy tolerance is achieved.  
Extensive experiments demonstrate our method’s superior mesh extraction capabilities compared to state-of-the-art approaches, particularly for models with intricate details and complex topologies.  

Our contributions include:  
\begin{itemize}
    \item A power diagram formulation that leverages SDF samples as weights, where incremental construction of the diagram and its dual tetrahedralization enable adaptive mesh extraction. 
    \item A point insertion strategy that optimally targets regions with the largest reconstruction discrepancy between the mesh and the continuous surface. 
    \item A progressive framework requiring only local updates per iteration, ensuring low computational cost. The algorithm terminates when the accuracy tolerance is met. 
\end{itemize}

\section{Related Work}
\label{sec:Related work}

Isosurface extraction from implicit fields is a fundamental problem in geometry processing with extensive applications. Existing approaches can be broadly categorized into two paradigms: operating on pre-defined grids and progressively evolving an initial surface estimate to match the target isosurface.

\subsection{Extraction on Pre-defined Grids}
The most widely used approach is the Marching Cubes \cite{10.1145/37401.37422}, which uniformly partitions the space into cubic elements and determines surface geometry by examining field values at cube corners using a carefully designed lookup table. The Marching Tetrahedrons method \cite{akio_doi__1991} extends this idea to tetrahedral elements, offering better flexibility in handling complex geometries. Dual contouring techniques \cite{10.1145/566654.566586} further enhance feature preservation by incorporating gradient information.

To improve efficiency, hierarchical data structures like Octree \cite{1372234,:10.2312/SGP/SGP07/125-133} can be used to increase resolution near surface regions and speed up computation. Although these approaches~\cite{6425062,shu1995adaptive} operate adaptively, they often use the signs at grid points to predict whether the surface intersects a grid cell, thus lacking a strong ability to detect unknown geometry.

The biggest advantage of these methods lies in their efficiency and parallelizability. However, they heavily rely on the resolution of pre-defined grid cells. Furthermore, locally refining the grid cells can easily cause mesh inconsistency between adjacent cells.

On this basis, recent data-driven methods, including Neural Marching Cubes \cite{DBLP:journals/corr/abs-2106-11272} and Neural Dual Contouring \cite{chen2022ndc}, leverage learned features to recover surface details without explicit gradient information. Another differentiable approach \cite{maruani2023voromeshlearningwatertightsurface} generates surfaces by optimizing point locations within Voronoi cells.

\subsection{Progressive Evolution Methods}
Progressive evolution methods offer an alternative paradigm. Early approaches in this category include surface tracking methods like Marching Triangles \cite{560840}, which grow meshes from seed points, and Delaunay refinement approaches \cite{10.1145/160985.161150,Wang2016OnVS} that incrementally improve mesh quality through careful vertex insertion.

This paradigm has gained significant attention recently due to its ability to iteratively correct surface topologies and refine geometric accuracy. Generally speaking, these methods utilize more information from the SDF to guide the refinement process. For instance, Reach for the Spheres \cite{sellán2023reachspherestangencyawaresurface} demonstrates that by exploiting sphere-surface tangency conditions, the surface can be better recovered even at low resolution. However, its fixed topology during optimization limits its applicability to complex models. Reach for the Arcs \cite{Sellan2024RFTA} addresses this topological limitation but requires repeatedly calling Poisson reconstruction, resulting in processing times orders of magnitude slower than conventional approaches, which makes it impractical for many applications.

McGrids \cite{renınst2024mcgridsmontecarlodrivenadaptive} offers another recent innovation in progressive extraction through adaptive Voronoi diagrams for incremental sampling. It separates the phase of SDF-guided sampling from the phase of surface extraction. In this paper, our algorithm alternately adds sample points and locally updates the discrete surface. The two operations are highly interrelated: the position of the newly added sample is guided by the current surface, while the surface is updated within the region where the new point is inserted.

\begin{figure}[h]
	\centering
\begin{overpic}
[width=.98\linewidth]{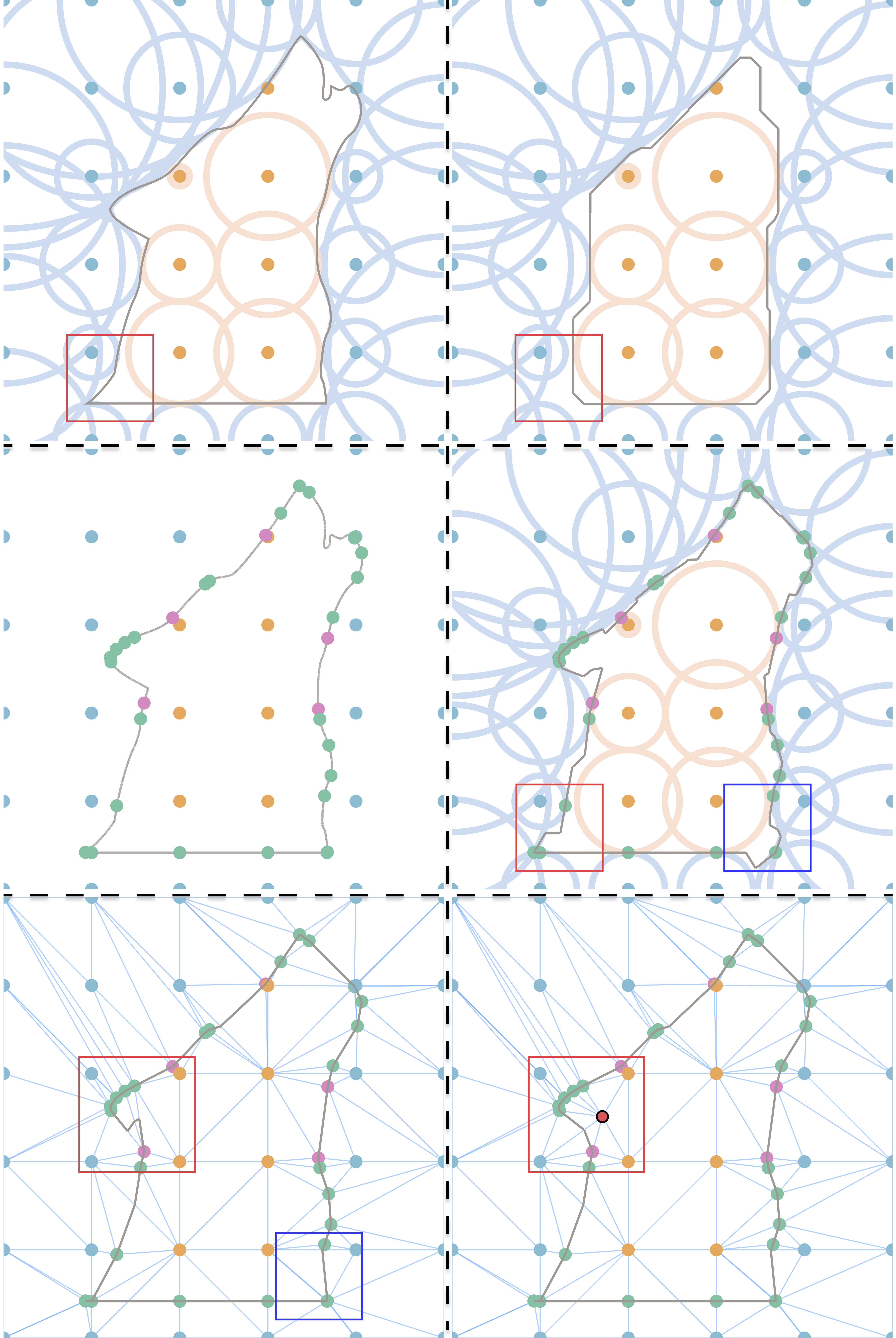}
\put(0,98){(a)}
\put(34,98){(b)}
\put(0,63.8){(c)}
\put(34,63.8){(d)}
\put(0,30){(e)}
\put(34,30){(f)}

\end{overpic}
\caption{
(a) Each sampling site generates a sphere with its center at the site and radius equal to its distance to the surface, creating tangential contact with the surface.
(b) The power diagram computed from these sites produces boundaries that lie between inner and outer spheres without intersections, but fails to maintain tangency with the spheres (highlighted in red).
(c) We incorporate projection points as additional power diagram sites, resulting in four distinct categories of sites.
(d) The enhanced power diagram computation achieves tangency with spheres (red box) but introduces geometric instabilities manifesting as zigzag patterns (blue box).
(e) By extracting the surface directly from the regular Delaunay, we eliminate the zigzag artifacts (blue box) while preserving the geometric properties.
(f) Our power diagram-based representation naturally supports incremental construction. When inserting a new point (shown in red), all structural changes in both the regular Delaunay and the reconstructed surface are confined within the local region (red box).
}
\label{fig:powerRepre}
  \vspace{-2mm}
\end{figure}

\section{Overview}
\label{sec:Overview}
Our method extracts isosurfaces from Signed Distance Fields using both gradient and value information as input. The algorithm consists of two primary components that work in a tightly coupled manner:

First, a power diagram-based isosurface extraction framework. Given spatial sampling points and their gradient information, this framework extracts a surface that approximates the underlying implicit geometry. It supports incremental insertion of sampling points and enables localized updates to the extracted surface, avoiding costly global recomputation.

Second, a novel point insertion strategy that prioritizes regions exhibiting the largest discrepancy between the current mesh and the underlying continuous surface. This strategy identifies optimal locations for new sampling sites that, when incorporated into the power diagram, effectively improve problematic regions in the current reconstruction.

These two components alternate in a continuous cycle—extraction, evaluation, refinement, and re-extraction—progressively enhancing surface fidelity until reaching either a predefined point budget or quality threshold. This framework implements a progressive approach that requires only local updates per iteration, thus minimizing computational overhead. The following sections provide detailed explanations of each component and their collaborative role in isosurface extraction algorithm.

\section{Power Diagram-based Isosurface Extraction}
\label{sec:PowerDiagramforSurfaceRepresentation}

We present our power diagram-based isosurface extraction approach in four progressive steps: First, we establish the theoretical foundation for isosurface extraction using power diagrams. Second, we demonstrate how incorporating projection points as additional sampling sites enhances extraction quality. Third, we employ the dual of the power diagram—regular Delaunay tetrahedralization—to overcome inherent instabilities in power diagram construction. Finally, we validate that this extraction framework supports localized incremental updates without requiring global recomputation.

\subsection{Theoretical Foundation}
RFTA \cite{Sellan2024RFTA} highlighted an essential geometric property: spheres constructed with centers at sampling points and radii equal to their unsigned surface distances maintain tangential contact with the surface. 
Figure~\ref{fig:powerRepre}(a) illustrates this property in the 2D case. In our framework, we classify sampling points as $\mathbf{P}_-$ and $\mathbf{P}_+$ based on their negative and positive signed distance values, forming inner spheres $\mathbf{S}_-$ and outer spheres $\mathbf{S}_+$ with radii equal to their unsigned distance values, as shown in Figure~\ref{fig:powerRepre}(a).
By utilizing these points as power diagram sites and their squared surface distances as weights, we establish the following theorem:

\begin{wrapfigure}[11]{r}{0.36\linewidth}
  \begin{center}
  \vspace{-10mm}
   \begin{overpic}[width=0.99\linewidth]{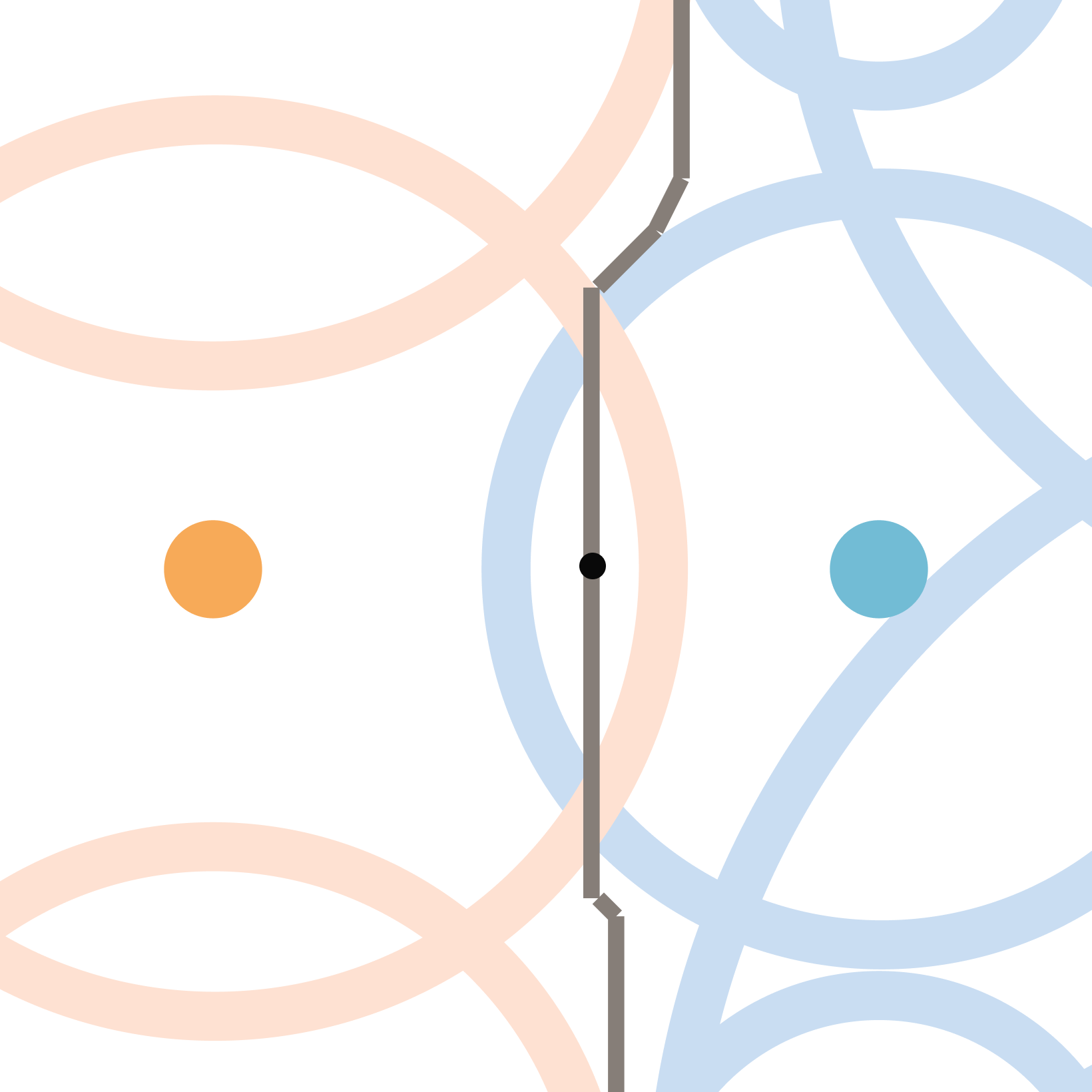}
      \put(49,50){$x$}             
      \put(10,52){$p^-_i$}           
      \put(70,52){$p^+_j$}           
    \end{overpic}
  \end{center}
\end{wrapfigure}
\begin{theorem}
For two sets of spheres $\mathbf{S}^-$ and $\mathbf{S}^+$, the boundary of the power diagram computed by using sphere centers as sites and squared distances as weights forms a valid surface between the two sets of spheres without intersecting any sphere.
\end{theorem}




\begin{proof} We prove by contradiction, as illustrated in the inset figure. Assume a power diagram boundary intersects with an inner sphere. This implies there exists a point $x$ on the power boundary inside this inner sphere, with a negative power distance to some inner site $p^-_i$: \begin{equation} |x - p^-_i|^2 - \phi^2(p^-_i) < 0 \end{equation}

Here, $p^-_i$ represents an inner point (the superscript ``$-$" denotes an interior point), $i$ is the index, and $\phi(p^-_i)$ is the distance from $p^-_i$ to the surface.

Since $x$ lies on the power diagram boundary, there must exist an outer site $p^+_j$ (where the superscript ``$+$" denotes an exterior point) with equal power distance: \begin{equation} |x - p^+_j|^2 - \phi^2(p^+_j) = |x - p^-_i|^2 - \phi^2(p^-_i) < 0 \end{equation}

This indicates that $x$ also lies inside the outer sphere centered at $p^+_j$, creating an intersection between an inner and outer sphere at point $x$, as shown in the inset. This contradicts the fundamental property that inner and outer spheres cannot intersect, thereby proving the theorem. \end{proof}

This theorem establishes that the power diagram boundary provides a feasible surface approximation by maintaining the crucial non-intersection property with both sets of spheres, as illustrated in a 2D example in Figure~\ref{fig:powerRepre}(b).
\subsection{Surface-Tangent Points Integration}
\label{sec:Surface-TangentPointGeneration}
Since SDF representation provides readily accessible gradient information, we can easily compute surface tangent points to enhance our reconstruction. For each sampling point $p_i$, we obtain its projection $q_i$ onto the zero-level surface as shown in the inset figure.

\begin{wrapfigure}{l}{0.3\linewidth}
  \begin{overpic}[width=1.02\linewidth]{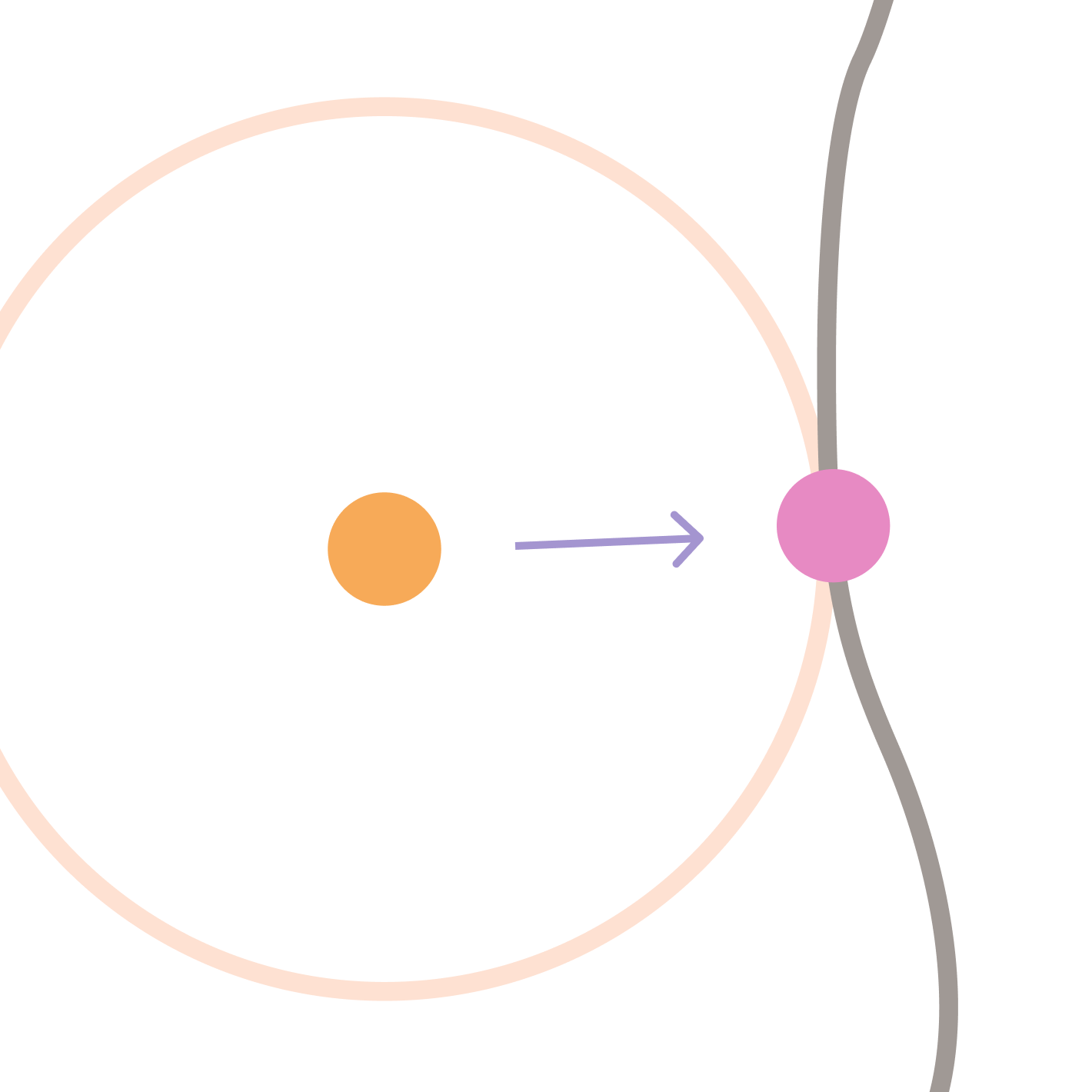}
    \put(18,55){$p_i$}
    \put(62,55){$q_i$}
  \end{overpic}
  \vspace{-15pt}
\end{wrapfigure}

We incorporate these projection points as additional power diagram sites alongside the original sampling points. Our power diagram sites are then classified into two categories, with $\mathbf{P}_{-}$ and $\mathbf{P}_{+}$ representing exterior and interior sampling points, and $\mathbf{P}_{-}^{\perp}$ and $\mathbf{P}_{+}^{\perp}$ denoting their respective surface projections:

\begin{itemize}
    \item Category 1: $\mathbf{P}_{-}$ and $\mathbf{P}_{+}^{\perp}$
    \item Category 2: $\mathbf{P}_{+}$ and $\mathbf{P}_{-}^{\perp}$
\end{itemize}

This classification creates power diagram boundaries that form planes tangent to the surface at the projection points. As a result, our reconstructed surface maintains proper geometric contact with the original surface at these projection points, significantly improving the fidelity of geometric features.

Figure~\ref{fig:powerRepre}(c) and (d) illustrate this point cloud classification and the corresponding power diagram extraction results in 2D scenarios.
\subsection{Stable Surface Extraction}
\label{sec:ModifiedDualityRulesforMeshConstruction}
While power diagram boundaries ensure non-intersection with inner and outer spheres and maintain surface tangency, they inherently produce geometric instabilities in practice, as highlighted in the blue box in Figure~\ref{fig:powerRepre}(d). Leveraging the dual relationship between power diagrams and regular Delaunay tetrahedralization, we propose extracting the surface mesh $\mathcal{M}$ directly from the regular Delaunay tetrahedralization $\mathcal{D}$ by processing each tetrahedron containing different types of sites. This approach preserves the spatial guidance properties of power diagrams while eliminating their inherent geometric discontinuities.

\begin{wrapfigure}[9]{r}{0.5\linewidth}
  \begin{center}
 \hspace{-8mm}
   \begin{overpic}[width=1.0\linewidth]{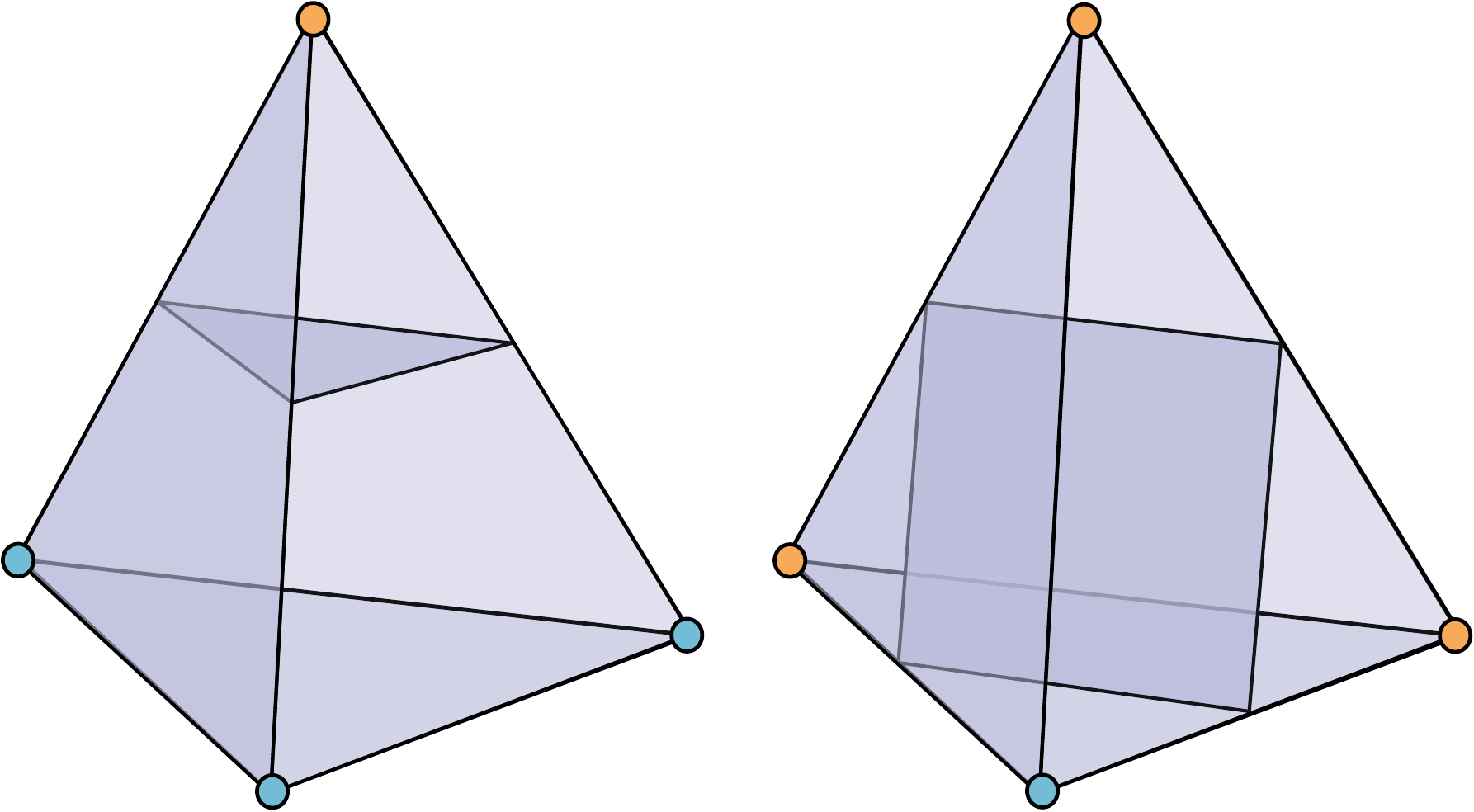}
    \end{overpic}
  \end{center}
\end{wrapfigure}
For each tetrahedron containing different types of sites, there are only two possible distributions: a 2-2 split or a 1-3 split, as shown in the inset. Consequently, the surface extracted within the tetrahedron is either triangular or quadrilateral.

For each edge connecting different types of sites, we extract a corresponding vertex using the following dualization rules, which depend on the properties of the two vertices forming the edge:

\begin{itemize}
\item $\begin{array}{ll}
v_1 \in \mathbf{P}{-} \text{ and } v_2 \in \mathbf{P}{-}^{\perp} & :v_2
\end{array}$
\item $\begin{array}{ll}
v_1 \in \mathbf{P}{+} \text{ and } v_2 \in \mathbf{P}{+}^{\perp} & :v_2
\end{array}$
\item $\begin{array}{ll}
v_1 \in \mathbf{P}{-} \text{ and } v_2 \in \mathbf{P}{+} & :\frac{v_1\phi(v_2)-v_2\phi(v_1)}{\phi(v_2)-\phi(v_1)}
\end{array}$
\item $\begin{array}{ll}
v_1 \in \mathbf{P}{-}^{\perp} \text{ and } v_2 \in \mathbf{P}{+}^{\perp} & :\frac{v_1+v_2}{2}
\end{array}$
\end{itemize}

By connecting these dual vertices according to the cell's topology, we construct a mesh containing both triangular and quadrilateral elements that approximates the surface while avoiding the instabilities inherent in power diagram dualization, as demonstrated in Figure~\ref{fig:powerRepre}(e) for a 2D case.

\subsection{Localized Incremental Update}
\label{sec:LocalizedIncrementalUpdate}
A significant advantage of our approach is its natural adaptability to incremental construction. The regular Delaunay tetrahedralization inherently supports localized updates, allowing us to efficiently process the insertion of new sites without global recomputation. When inserting a new site, only a small neighborhood of existing elements requires modification, and correspondingly, only the associated portions of the dual mesh need to be updated.

This locality property ensures that the computational complexity of each update operation depends only on the size of the affected region, not on the total mesh complexity. As illustrated in Figure~\ref{fig:powerRepre}(f), when inserting a new point (marked in red) in a 2D scenario, all updates are confined within the highlighted red box. This localized behavior makes our approach particularly efficient for adaptive refinement and dynamic surface reconstruction tasks.

\begin{figure*}[h]
	\centering
\begin{overpic}
[width=.98\linewidth]{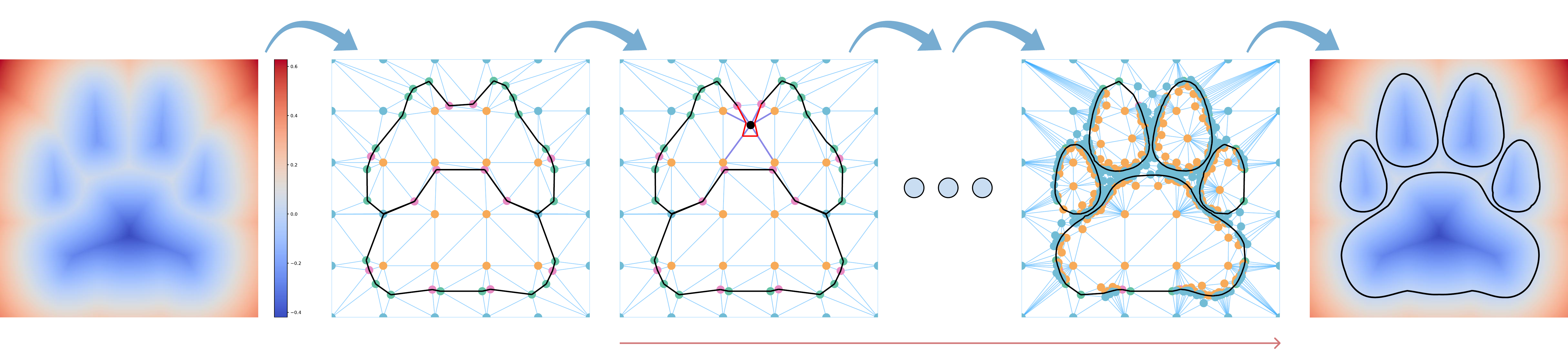}
\put(18.5,21.5){init}
\put(36,21.5){+site}
\put(55,21.5){+site}
\put(62,21.5){+site}
\put(79,21.5){extract}

\put(6,-0.5){Input}
\put(24.5,-0.5){Initialization}
\put(57,-1.5){Evolution}
\put(90,-0.5){Result}

\end{overpic}
\vspace{3mm}
\caption{
Given a signed distance field with associated gradients, we initialize the surface by constructing a regular Delaunay tetrahedralization from uniform samples and their surface projections. The algorithm then iteratively identifies suboptimal geometric elements (triangles in 3D, segments in 2D) and refines the reconstruction through strategic site insertion and incremental Delaunay updates until reaching either the point count limit or convergence threshold. 
}
\label{fig:pipeline}
\end{figure*}

\section{Algotihm}
\label{sec:algorithm}
Starting from a surface extracted from uniformly distributed sampling points, we iteratively refine the surface by identifying regions with suboptimal reconstruction quality and strategically adding new points in those areas until convergence. Since our key algorithmic components are not executed sequentially but rather form a tightly coupled iterative process, we first present the overall pipeline before detailing each critical step.

\subsection{Pipeline}
Our algorithm proceeds through the following steps:

\noindent{\bf Step 1.} Construct regular delaunay~$\mathcal{D}$ from uniform grid vertices and their projections, obtain initial surface $\mathcal{M}$ ~(Section~\ref{sec:PowerDiagramforSurfaceRepresentation}).

\noindent{\bf Step 2.} Calculate the deviation metric $\delta$ between each triangle in $\mathcal{M}$ and the underlying surface~(Section~\ref{sec:TriangleQualityAssessment}) and organize them in max-heap $\mathcal{Q}$.

\noindent{\bf Step 3.} Extract triangle with maximum deviation $\delta_c$ from $\mathcal{Q}$, add corresponding new site $p_c$ and its projection $q_c$~(Section~\ref{sec:pointdetermine}). Update $\mathcal{D}$ and $\mathcal{M}$~(Section~\ref{sec:LocalizedIncrementalUpdate}), compute quality metrics for affected and newly added faces (Section~\ref{sec:TriangleQualityAssessment}) and insert them into $\mathcal{Q}$.

\noindent{\bf Step 4.} If the number of inserted points reaches $k_{max}$ or $\delta_c < \epsilon$, proceed to Step 5. Else return to Step 3.

\noindent{\bf Step 5.} Extract surface.






Each key step in our pipeline is detailed in the corresponding subsections referenced above. For better understanding, we also provide a 2D illustration of this pipeline in Figure~\ref{fig:pipeline}.

\begin{figure}[h]
	\centering
\begin{overpic}
[width=.7\linewidth]{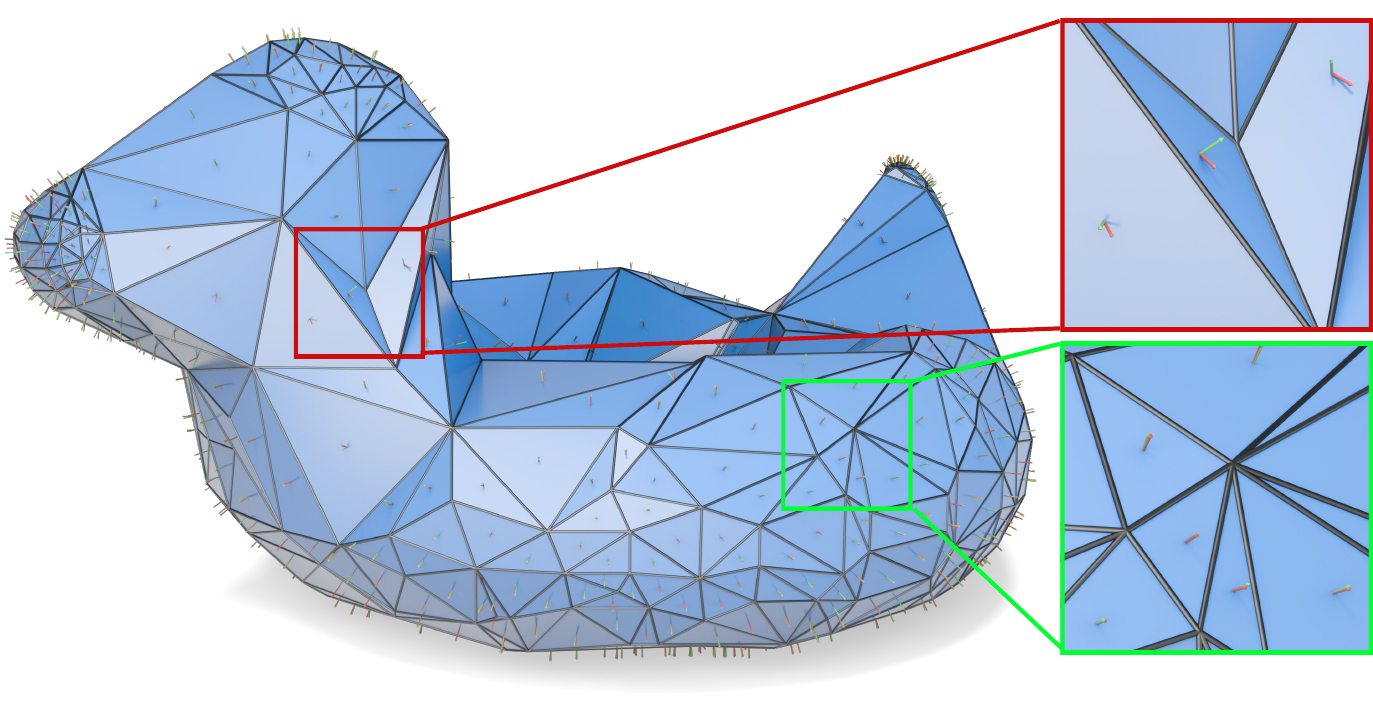}
\end{overpic}
\caption{
When a face is well-aligned with the zero-level set, the SDF gradient directions throughout its interior closely match the face normal (green box). Conversely, faces deviating significantly from the zero-level set exhibit interior points where gradient directions substantially differ from the face normal (red box).
}
\label{fig:triangleAsses}
  \vspace{-5mm}
\end{figure}

\subsection{Isosurface Approximation Quality Evaluation}
\label{sec:TriangleQualityAssessment}
A robust metric for assessing triangle fidelity to the zero-level set is crucial for our adaptive refinement strategy. Our approach leverages the geometric principle that for a triangle perfectly aligned with the isosurface, the SDF gradient $\nabla \phi(\mathbf{x})$ should parallel the triangle normal $\hat{\mathbf{n}}$ at any point $\mathbf{x}$ on the triangle, as illustrated in Figure~\ref{fig:triangleAsses}.

To quantify deviation from this ideal alignment, we measure the orthogonal component of the gradient for $\triangle$:
\begin{equation}
\delta(\triangle) = \int_{\triangle} |\nabla \phi(\mathbf{x}) - (\nabla \phi(\mathbf{x}) \cdot \hat{\mathbf{n}}) \hat{\mathbf{n}}| \, dA
\end{equation}

We select this gradient-based metric over simple SDF values for its enhanced sensitivity to topological variations through normal field analysis, as demonstrated in Figure~\ref{fig:triangleAsses}.

For practical implementation, we approximate the integral by partitioning the triangle into three sub-triangles using angle bisectors and assuming constant gradient within each:
\begin{equation}
\delta(\triangle) \approx \sum_{i=1}^3 A_i |\nabla \phi(\mathbf{c}_i) - (\nabla \phi(\mathbf{c}_i) \cdot \hat{\mathbf{n}}) \hat{\mathbf{n}}|
\end{equation}
where $A_i$ is the area of the $i$-th sub-triangle and $\mathbf{c}_i$ its centroid. This discretization provides an efficient and accurate quality metric for guiding point insertion.

It is worth noting that this gradient-alignment metric might theoretically yield a favorable quality score for a normal-aligned triangle that is positioned at a distance from the actual zero-level set. However, our algorithm incorporates two complementary mechanisms that prevent such cases: (1) Our surface projection step consistently maintains triangles in close proximity to the zero-level set through the use of surface-tangent points, and (2) In the unlikely event that such a misplaced triangle exists, adjacent triangles would exhibit high energy values, which would automatically trigger point insertion during the refinement process to correct the local geometry.

\subsection{Strategic Site Insertion for Local Refinement}
\label{sec:pointdetermine}

After identifying regions with poor approximation quality through our deviation metric $\delta$, we strategically place new sites to improve these problematic areas. For each triangle with high deviation, we select the dual vertex of its corresponding tetrahedron in the regular Delaunay tetrahedralization $\mathcal{D}$ and its surface projection as new insertion sites.

\begin{wrapfigure}[11]{r}{0.6\linewidth}
  \begin{center}
    \includegraphics[width=0.97\linewidth]{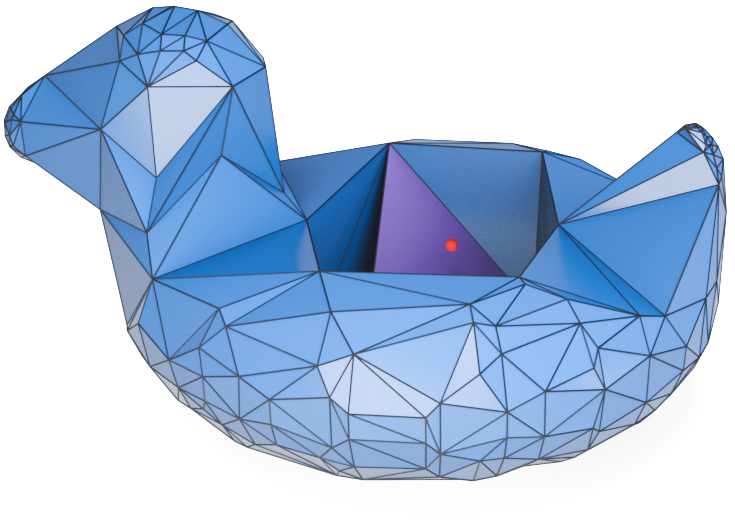}
  \end{center}
\end{wrapfigure}

As shown in the inset figure, the purple tetrahedron generates the problematic triangle, and the red point represents its dual vertex in the power diagram. We add both this point and its surface projection to the regular Delaunay construction to enhance the local reconstruction quality.

This approach is particularly effective because when significant inconsistency is detected between a triangle in $\mathcal{M}$ and the zero-level set, the inserted sites must modify the local topology of precisely that region. The dual vertex serves as an ideal candidate for this purpose due to the one-to-one correspondence between triangles in $\mathcal{M}$ and tetrahedra in $\mathcal{D}$.

\begin{figure*}[h]
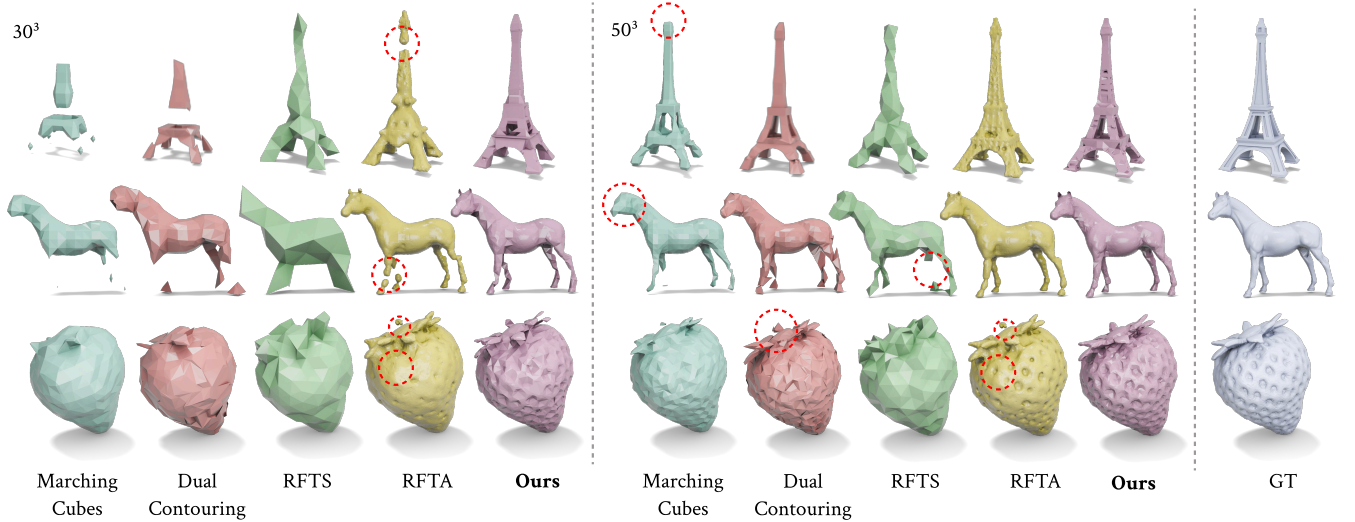

	\centering
\begin{overpic}
[width=.98\linewidth]{sec/figures/input3comparasion.pdf}
\end{overpic}
\caption{
Comparison with state-of-the-art methods on three low-genus inputs. Results are generated using uniform $30^3$ and $50^3$ grid samples. Our method demonstrates superior expressiveness in capturing geometric details.
}
\label{fig:comprasionWithOther}
  \vspace{-2mm}
\end{figure*}
\begin{figure*}[h]
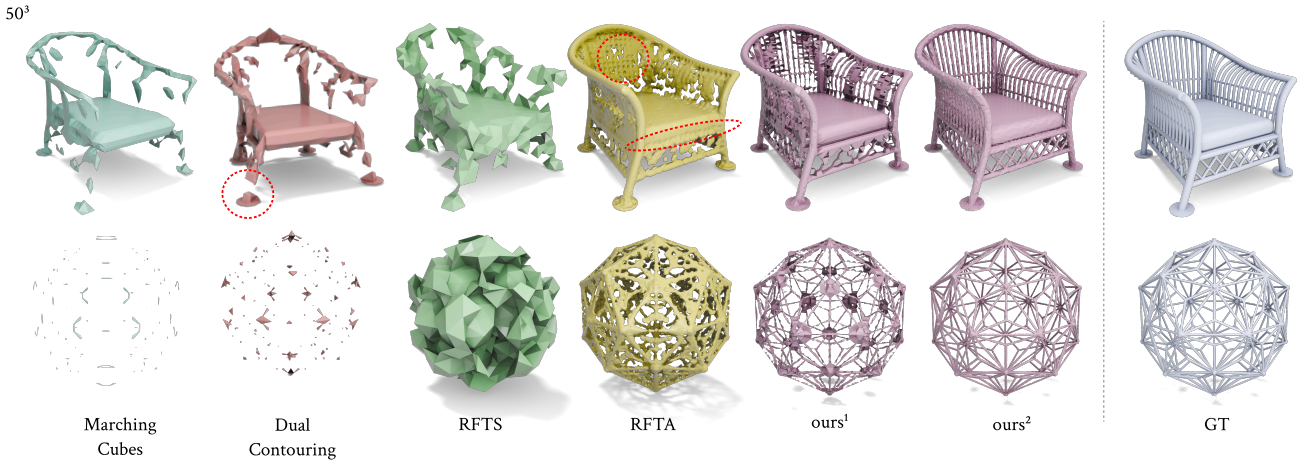

	\centering
\begin{overpic}
[width=.95\linewidth]{sec/figures/high-genus.pdf}
\end{overpic}
\caption{
Comparison with state-of-the-art methods on two high-genus inputs. Previous methods and Ours$^1$ use uniform $50^3$ grid samples, while Ours$^2$ starts from an $8^3$ grid and applies our adaptive sampling strategy until reaching 62500 samples. Our method demonstrates better topological preservation, with further improvements achieved through our adaptive sampling strategy.
}
\label{fig:comparasionWithHighComplex}
  \vspace{-2mm}
\end{figure*}

\section{Evaluation}
\subsection{Experiment Setting}
Our experimental evaluation was conducted on a system equipped with an AMD Ryzen 9 5950X CPU and 72 GB RAM. The test dataset comprised 100 models curated from Thingi10k and models previously used in~\cite{Sellan2024RFTA} and~\cite{sellán2023reachspherestangencyawaresurface}, supplemented with some custom-designed models. 
Following RFTS's~\cite{sellán2023reachspherestangencyawaresurface} experimental design, all models were normalized to fit within a unit cube of $[-1/2, 1/2]^3$, while the surface extraction was performed in the extended domain of $[-1, 1]^3$.
Our comparative experiments focused on two principal scenarios: (1) extraction under fixed input conditions (the core application scenario of RFTS/RFTA) to validate extraction capability, and (2) extraction in continuous fields.

\vspace{3mm}
\noindent
\textbf{Comparison Methods.}
Our experimental evaluation benchmarks the proposed approach against five representative isosurface extraction techniques from the literature: three recent techniques—RFTA (REACH for Arc)~\cite{Sellan2024RFTA}, RFTS (REACH for Sphere)~\cite{sellán2023reachspherestangencyawaresurface}, and MCGrids~\cite{renınst2024mcgridsmontecarlodrivenadaptive}—and two classical methods—Marching Cubes~\cite{10.1145/37402.37422} and Dual Contouring~\cite{10.1145/566654.566586}.

RFTA and RFTS utilize optimization-based approaches that exploit the tangential relationship between sample-centered spheres and surfaces for geometric reconstruction. Here, each sphere is centered at a sample point with its radius defined by the distance to the surface. For RFTS, we employed the default parameters provided in the authors' GPyToolbox implementation~\cite{gpytoolbox}, while RFTA was configured with 20 fine-tuning iterations and 30 maximum points per sphere. MCGrids employs an iterative isosurface extraction approach, whereas Marching Cubes and Dual Contouring implement uniform spatial discretization strategies for surface extraction.

We categorize existing methods into two distinct classes: Direct Extraction Methods, which utilize fixed information for isosurface extraction without additional input during reconstruction (including Marching Cubes, Dual Contouring, RFTS, and RFTA), and Iterative Refinement Methods, which progressively insert new points and adjust results during the extraction process (McGrids).

\subsection{Comparison with Direct Extraction Methods}
Figure~\ref{fig:comprasionWithOther} illustrates comparative analyses at uniform grid resolutions of $30^3$ and $50^3$ over the $[-1,1]^3$ domain, where all algorithms using identical grid point SDF values and gradient information for low-genus model reconstruction.
Marching Cubes and Dual Contouring utilize only local SDF information, resulting in poor detail preservation; 
RFTS cannot handle topological variations, restricting its application to simple models; and RFTA, despite improving upon its predecessor, fails to recover intricate surface details. Our method leverages projection points derived from gradient information with regular Delaunay-based extraction to achieve superior reconstruction quality. This is demonstrated in Figure \ref{fig:comprasionWithOther} with the strawberry model and in Figure \ref{fig:golf} with the golf ball model, where our approach successfully reconstructs surface cavities that RFTA fails to capture. Additionally, at $50^3$ resolution, RFTA's tower reconstruction exhibits notable spherical artifacts.

For high-genus models (Figure \ref{fig:comparasionWithHighComplex}), we present two implementations: Our$^1$, which uses standard $50^3$ grid inputs, and Our$^2$, which employs an adaptive approach starting from an $8^3$ grid and expanding to $50^3$ total points (including both sample points and their projections) using the strategy detailed in Section \ref{sec:algorithm}. While conventional methods (Marching Cubes, Dual Contouring, RFTS) fail to preserve intricate details such as chair backs, and RFTA shows only modest improvement, both of our implementations demonstrate better detail preservation. Our$^2$, through its adaptive point sampling strategy, achieves particularly significant improvements in detail preservation.
\begin{figure}[h]
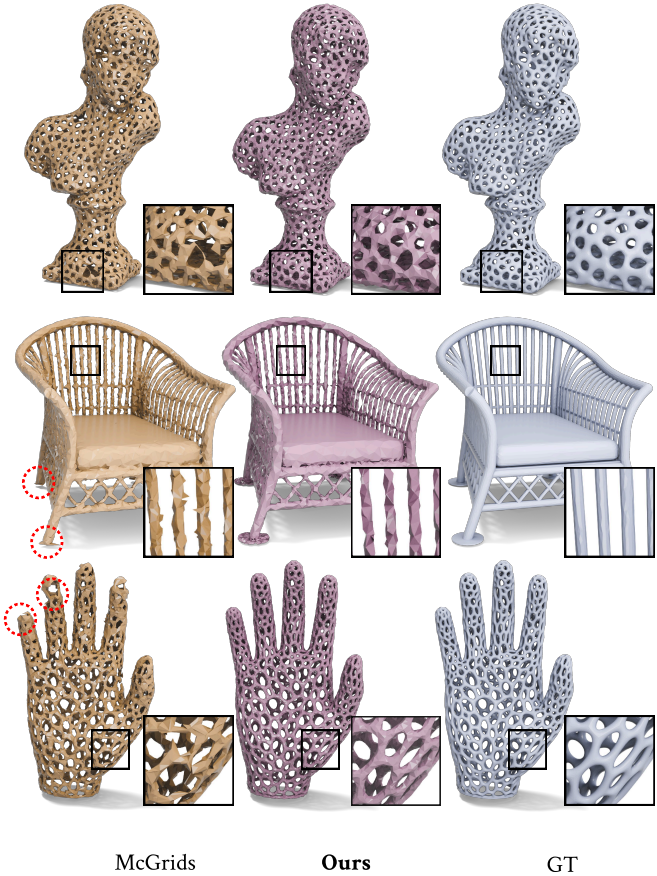

	\centering
\begin{overpic}
[width=.99\linewidth]{sec/figures/figfinger.pdf}
\end{overpic}
\caption{
Comparison with McGrids~\cite{renınst2024mcgridsmontecarlodrivenadaptive}. Both methods employ an almost identical number of sampling points. Our method exhibits superior preservation of fine details and topological features.
}
\label{fig:comparasionWithMcGrids}
  \vspace{-2mm}
\end{figure}

\subsection{Comparison with Progressive Refinement Methods}
McGrids introduces a Monte Carlo-based strategy that formulates adaptive grid construction as a probability sampling problem, employing Delaunay tetrahedralization followed by marching tetrahedra for isosurface extraction. The method implements a two-phase point addition strategy: initial Monte Carlo sampling followed by refinement.
In comparative analysis (Figure~\ref{fig:comparasionWithMcGrids}) using approximately 20,000 sampling points for both methods, we achieve better results.

Our method demonstrates superior topology preservation for complex models, primarily due to our more comprehensive utilization of available information. While McGrids relies solely on SDF sign changes for extraction, our approach leverages both distance values and gradient information through regular Delaunay tetrahedralization, enabling more accurate surface inference and significantly enhancing extraction capabilities.

Furthermore, although MCGrids implements refinement capabilities in its second phase, it fails to detect extraction errors in extreme cases, such as models with radii smaller than its second-phase threshold. In contrast, our normal-based detection mechanism successfully identifies such cases. 

\begin{table}
\caption{Quantitative comparison (CD $\cdot 10^{-5}$) for fixed grid inputs.}
\vspace{-3mm}
\centering\small
\resizebox{0.8\columnwidth}{!}{%
  \renewcommand{\arraystretch}{1.25}
  \setlength{\tabcolsep}{7pt}
  \rowcolors{2}{white}{rowgray}
  \begin{tabular}{c|ccccc}
  \hline
  \rowcolor{headerblue} Resolution & MC & DC & RFTS & RFTA & Ours$^1$ \\
  \hline
  30$^3$ & 267 & 252 & 244 & 11.5 & \textbf{2.68} \\
  40$^3$ & 140 & 114 & 258 & 8.27 & \textbf{1.30} \\
  60$^3$ & 68.4 & 60.9 & 236 & 4.93 & \textbf{0.48} \\
  \hline
  \end{tabular}%
}
\vspace{-2mm}
\label{table:reso}
\end{table}

\begin{table}
\caption{Quantitative comparison on complex models.}
\centering\small
\resizebox{\columnwidth}{!}{%
  \renewcommand{\arraystretch}{1.25}
  \setlength{\tabcolsep}{7pt}
  \rowcolors{2}{white}{rowgray}
  \begin{tabular}{c|cccccc}
  \hline
  \rowcolor{headerblue}  & CD$\cdot 10^5$ $\downarrow$ & NC $\uparrow$ & ECD$\cdot 10^3$ $\downarrow$ & F1 $\uparrow$ & EF1 $\uparrow$ & Time (s) $\downarrow$\\
  \hline
  MC & 0.571 & 0.869 & 0.359 & 0.974 & 0.174 & 11.46 \\
  DC & 0.864 & 0.794 & 0.352 & 0.945 & 0.167 & 61.63 \\
  RFTS & 3.154 & 0.761 & 9.450 & 0.636 & 0.009 & 34.17 \\
  RFTA & 3.050 & 0.836 & 1.485 & 0.729 & 0.007 & $>$ 1h \\
  McGrids & 0.264 & 0.888 & 0.379 & \textbf{0.999} & \textbf{0.191} & 10.53 \\
  Ours$^2$ & \textbf{0.216} & \textbf{0.897} & \textbf{0.348} & \textbf{0.999} & 0.187 & \textbf{10.15} \\
  
  \hline
  \end{tabular}%
}
\vspace{-5mm}
\label{table:complex}
\end{table}
\subsection{Quantitative Evaluation}
To comprehensively evaluate our method's effectiveness, we conducted extensive quantitative comparisons with existing approaches. Our quantitative evaluation comprised two distinct experiments. First, we assessed extraction accuracy with fixed grid inputs—the core scenario addressed by RFTS/RFTA—at three different resolutions across 100 diverse models, with $L_2$ Chamfer Distance (CD) results presented in Table~\ref{table:reso}. Second, we compared extraction quality on 50 complex shapes with intricate geometric features, where MC, DC, RFTS, and RFTA utilized 200$^3$ uniform grids, while McGrids and our proposed method employed 60K adaptive samples, as shown in Table~\ref{table:complex}. 
For all quantitative evaluations, we used 1 million sampled points to calculate the error metrics.
The results demonstrate that our method achieves state-of-the-art extraction quality under fixed uniform sampling conditions while also producing superior surface reconstruction results in minimal computational time when utilizing adaptive SDF sampling strategies.

\begin{figure}[h]
	\centering
\hspace{-2mm}
\begin{overpic}
[width=1\linewidth]{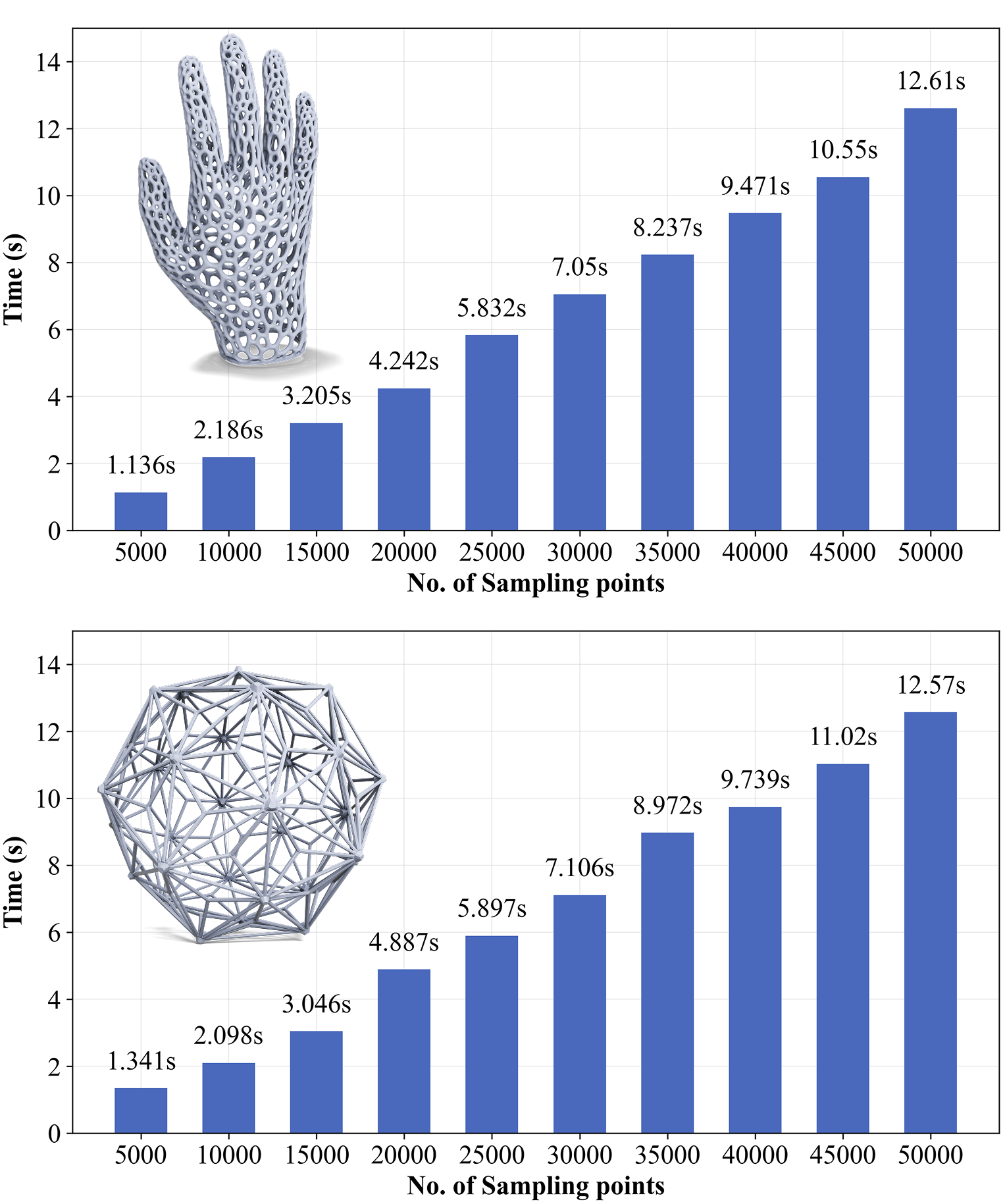}
\end{overpic}
\caption{
Performance scaling of our algorithm on two complex models with increasing sampling points. Benefiting from our localized update strategy, the computation time exhibits approximately linear growth with respect to the number of sampling points.}
\label{fig:time}
\end{figure}

\begin{figure*}[h]
	\centering
\begin{overpic}
[width=0.99\linewidth]{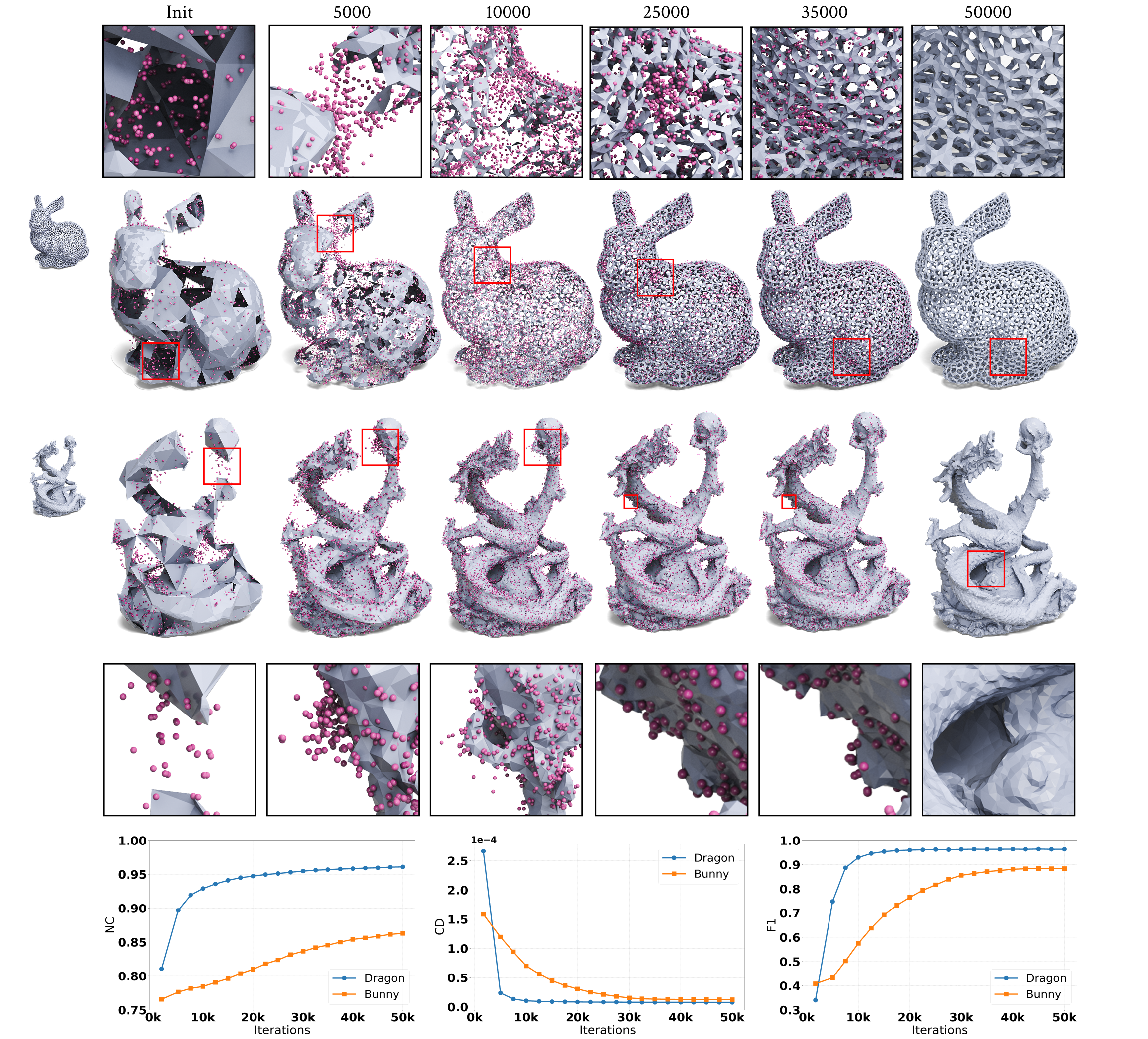}
\end{overpic}
\caption{
Intermediate stages of our algorithm demonstrated on two classic models with high genus and rich surface details. The process illustrates two key aspects of our site placement strategy: (1) Topology correction - shown in the first model, where new sites are strategically added to either establish necessary connections or separate incorrectly merged regions; (2) Geometric refinement - demonstrated in the second model, where progressive site addition continuously enriches surface details.
}
\label{fig:IterativeEvolutionAnalysis}
  \vspace{-2mm}
\end{figure*}

\subsection{Computational Efficiency Analysis} Figure~\ref{fig:time} presents a analysis of our algorithm's computational performance, measured by evaluating execution time against the number of iteratively inserted sampling points on two topologically complex models. Our localized update strategy demonstrates significant efficiency advantages: the time required to insert each new point remains approximately constant regardless of the total number of points already incorporated into the regular Delaunay tetrahedralization. This locality property enables our method to achieve computational complexity that scales approximately linearly with point count ($O(n)$), avoiding the higher-order complexity typically associated with global mesh reconstruction approaches. The empirical measurements confirm that our incremental framework maintains consistent performance even as the sampling density increases, validating the effectiveness of our design decision to prioritize localized updates over global recomputation during the progressive refinement process.

\subsection{Progressive Topological and Geometric Convergence Analysis}
Figure~\ref{fig:IterativeEvolutionAnalysis} illustrates the progressive refinement of our algorithm on two representative models: one with high-genus topology and another with intricate surface details. Beginning with a uniform initialization of $15^3$ sampling points, the algorithm sequentially incorporates new sites until reaching the target of 50,000 sites. At each intermediate stage, highlighted sites indicate newly added positions in the transition to the subsequent refinement level. The visualization reveals two fundamental mechanisms of our adaptive site insertion strategy: (1) Topology correction—demonstrated in the high-genus model, where sites are strategically positioned to either establish critical connections or separate incorrectly merged regions; and (2) Geometric refinement—evidenced in the detailed model, where the progressive addition of sites systematically enhances surface feature resolution. These results confirm the efficacy of our quality-driven insertion approach in simultaneously improving topological accuracy and geometric fidelity through targeted, iterative refinement.

\subsection{Ablation Study}
To systematically evaluate the critical components of our method, we conducted comprehensive ablation studies on two fundamental aspects: (1) the contribution of projection points in the isosurface extraction process, and (2) the efficacy of our quality-driven adaptive point insertion strategy.
 
\begin{figure}[h]
	\centering
\begin{overpic}
[width=1\linewidth]{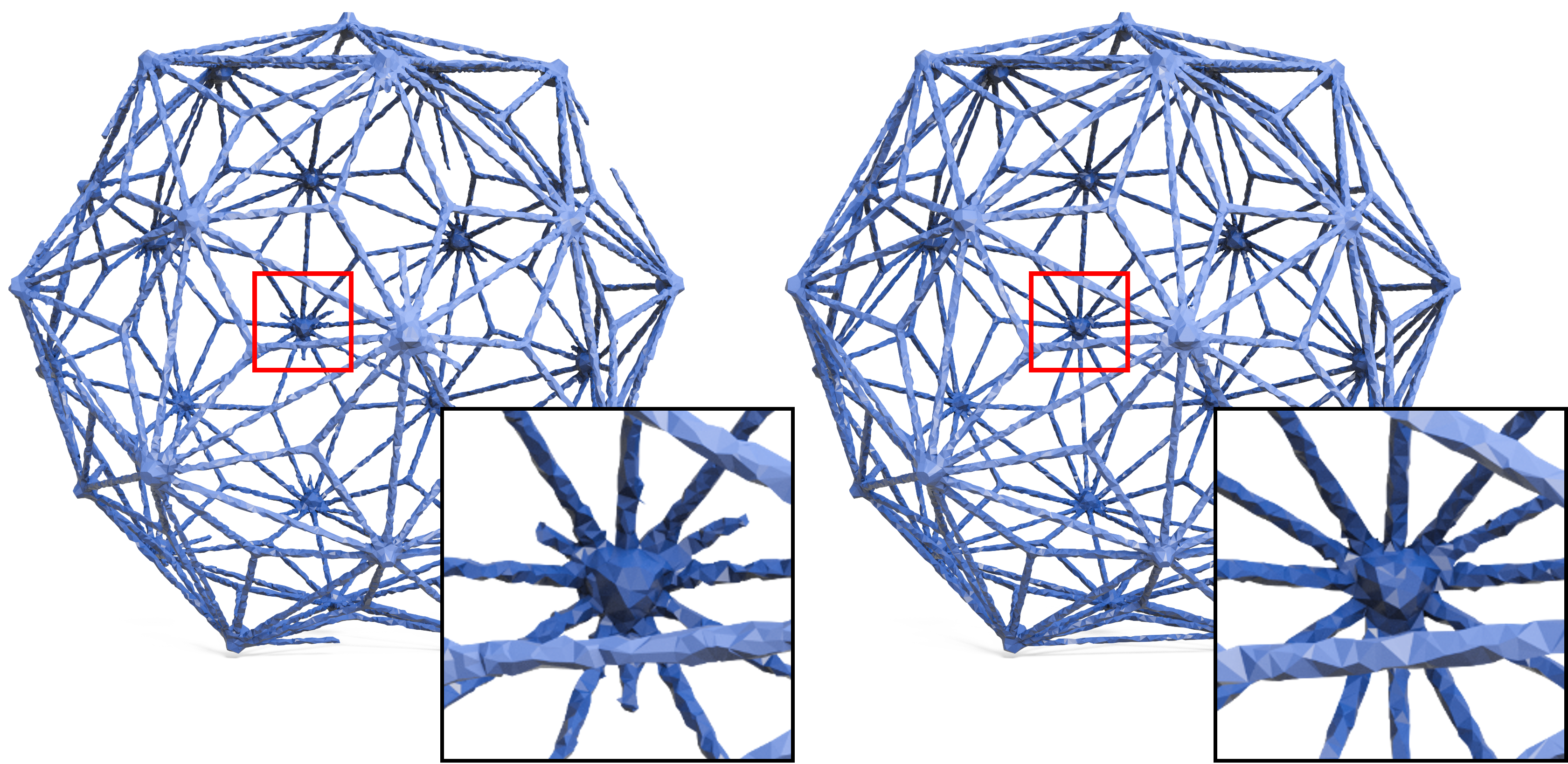}
\put(2,2){(a)}
\put(60,2){(b)}
\end{overpic}
\caption{Comparison of results using (a) sampling points only versus (b) our complete method with both sampling points and their surface projections. Starting from a uniform $20^3$ grid and refining to 40,000 points, the projection-enhanced approach demonstrates superior topological preservation.}
\label{fig:project}
\end{figure}

\textbf{Impact of SDF Projection Points.}
Incorporating projection points into our power diagram computation significantly improves the algorithm's capacity to accurately capture the underlying surface during isosurface extraction, particularly in regions with complex topological structures or fine geometric details. Figure~\ref{fig:project} illustrates the comparative results using a uniform $15^3$ grid as initial input, iteratively refined to 40,000 sampling points with: (a) our baseline method using only sampling points without projections, and (b) our complete method integrating both sampling points and their corresponding surface projections.

The visual comparison reveals that the projection-enhanced approach successfully preserves topological integrity across the entire structure, while the baseline method exhibits substantial topological defects and missing features. These results demonstrate that surface projections provide crucial geometric cues that guide the power diagram construction toward more accurate surface representation. This guidance is especially valuable for reconstructing challenging features such as thin structures, sharp edges, and regions with high curvature variation that might otherwise be inadequately captured when relying solely on discrete SDF sampling points.

\begin{figure}[h]
	\centering
\begin{overpic}
[width=1\linewidth]{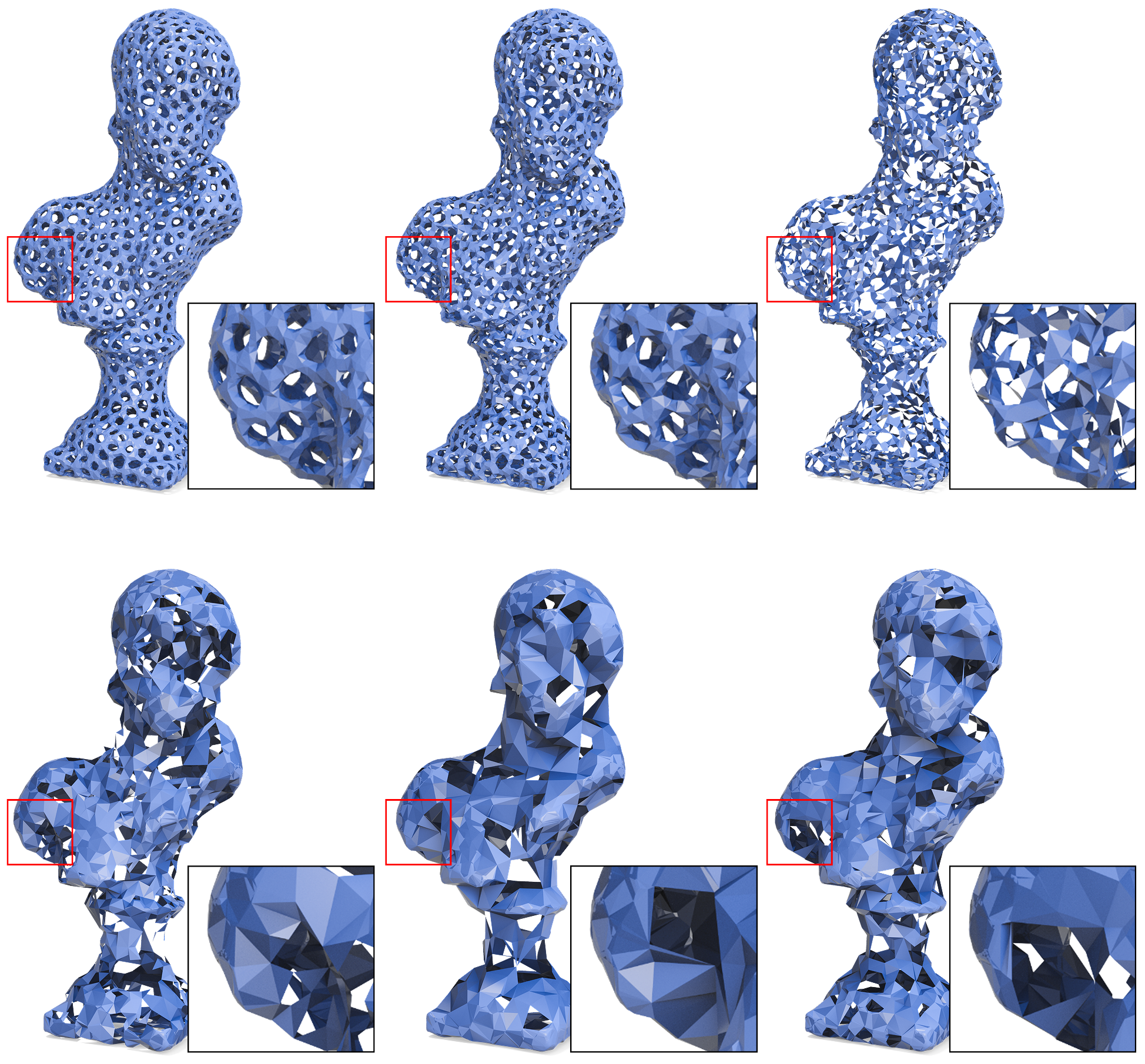}
\put(10,45){$5^3$}
\put(43,45){$33^3$}
\put(77,45){$34^3$}

\put(10,-4){$35^3$}
\put(43,-4){$40^3$}
\put(77,-4){$43^3$}
\end{overpic}
\vspace{2mm}
\caption{
Extraction results using different initial uniform sampling densities while maintaining 40,000 total sampling points. When more points are determined by our quality-driven strategy rather than uniform initialization, topological preservation significantly improves, confirming the effectiveness of our adaptive refinement approach.}
\label{fig:vary}
\end{figure}

\textbf{Effect of Adaptive Point Insertion Strategy on Topology Preservation.}
To quantify the effectiveness of our sequential point insertion strategy, we conducted a comparative analysis using varying initial uniform grid resolutions while maintaining a consistent total of 40,000 sampling points. As illustrated in Figure~\ref{fig:vary}, reconstruction quality correlates directly with the proportion of adaptively inserted samples: when a higher percentage of points are strategically added through our quality-driven criteria, the model demonstrates significantly superior topological fidelity compared to reconstructions relying predominantly on uniform sampling. The uniform sampling approach produces reconstructions with notable topological artifacts and missing features, particularly in geometrically complex regions. These results empirically validate that our adaptive point insertion mechanism intelligently prioritizes topologically critical regions, optimizing computational resources to preserve essential geometric features with minimal sampling budget.

\subsection{Neural SDF Extraction}
Neural SDFs have gained widespread applications in computer graphics and vision. Although neural SDFs cannot strictly guarantee unit-length gradients, making projection points inaccurate, our method demonstrates robust extraction capabilities when applied to these representations. We validated our approach on two classic neural SDF frameworks: SIREN-based SDFs fitted from mesh data and VolSDF reconstructed from multi-view images. As shown in Figure~\ref{fig:ndf}, even with only 10,000 sampling points, our method successfully extracts complex structures from SIREN SDFs, demonstrating that our algorithm maintains its powerful topological preservation capabilities when applied to neural implicit representations.

Efficient extraction from neural SDFs requires effective parallelization to address computational overhead. We implemented a batch-based processing strategy where multiple points are inserted in sequence, followed by parallel evaluation of all affected regions before proceeding to the next insertion batch. This approach significantly reduces neural network query overhead by leveraging GPU parallelization, making our method computationally practical for neural SDF extraction tasks.

The effectiveness of our extraction strategy on neural SDFs stems from two complementary factors. First, our adaptive sampling mechanism prioritizes point placement near the zero-level set—precisely where neural SDFs tend to maintain higher accuracy. Second, the approximation errors in neural SDFs often exhibit a self-balancing property, where local inaccuracies in different regions tend to compensate for each other, resulting in globally coherent reconstructions when processed through our power diagram framework.
\begin{figure}[h]
\centering
\begin{overpic}[width=0.9\linewidth]{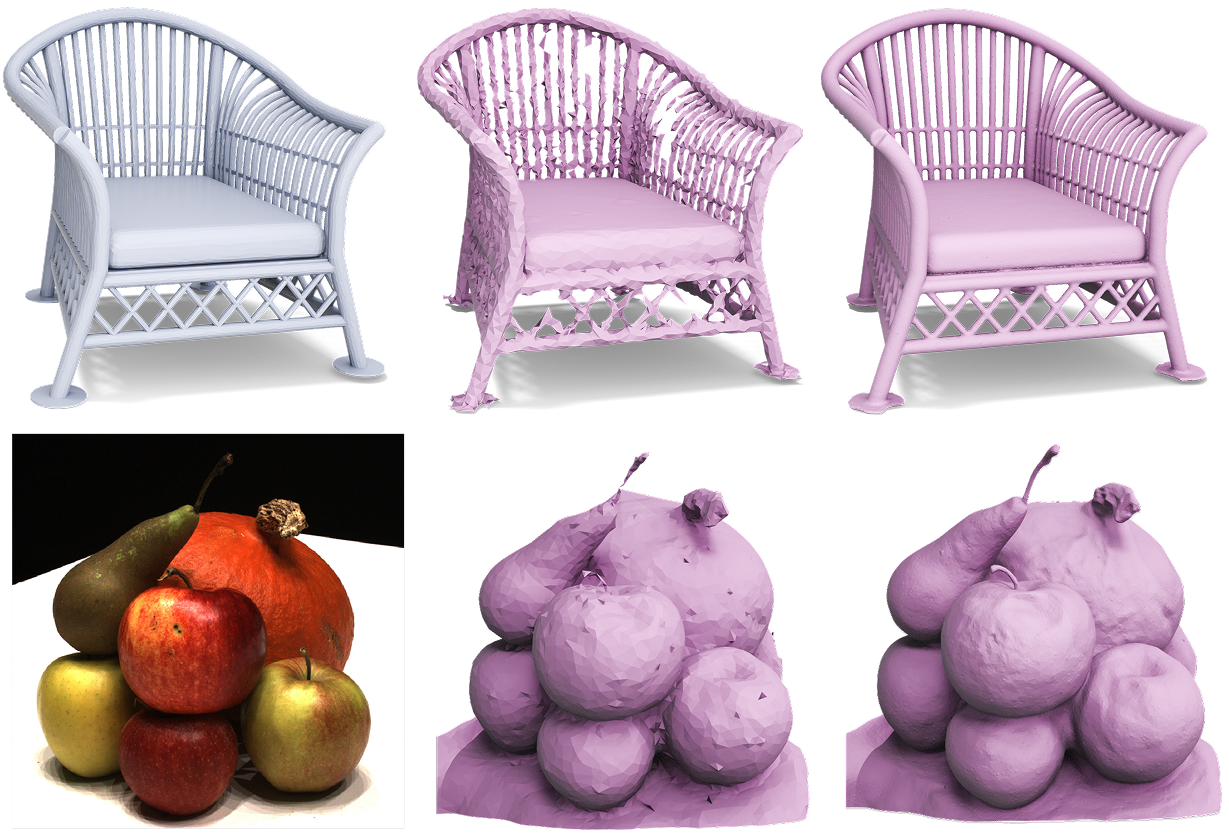}
    \put(9,68){\textbf{GT}}
    \put(32,68){\textbf{1w Sampling}}
    \put(65,68){\textbf{Dense Sampling}}
    \put(-5,44){\rotatebox{90}{\textbf{SIREN}}}
    \put(-5,10){\rotatebox{90}{\textbf{VolSDF}}}
\end{overpic}
\vspace{-1mm}
\caption{Neural SDF extraction in two representative scenarios: (1) SIREN-based SDFs fitted from mesh data and (2) VolSDF reconstructed from multi-view images.}
\vspace{-3mm}
\label{fig:ndf}
\vspace{-1mm}
\end{figure}

\section{Conclusion, Limitations and Future Work}

In this paper, we have presented a novel isosurface extraction method that advances capabilities in two key ways: first, by leveraging both SDF gradients and values with regular Delaunay tetrahedralization to utilize comprehensive information rather than merely sign changes; second, by introducing an iterative framework that couples extraction with adaptive refinement, guiding strategic point insertion based on quality assessment. This approach improves surface fidelity through localized updates without global recomputation. Experiments validate our method's effectiveness across diverse geometric complexities, showing superior performance for models with intricate details, thin structures, and complex topologies.

Our method has two primary limitations. First, it requires SDF gradient information during execution, which may be unavailable in certain contexts. Second, while our approach achieves dynamic isosurface extraction by maximizing information utilization, our surface quality assessment process necessitates multiple SDF queries without specific optimizations to reduce query count. This may incur significant computational overhead in certain scenarios.

For future work, we aim to address these limitations in two directions. We plan to develop techniques for estimating tangent points based on power diagrams, eliminating the explicit requirement for SDF gradient information. Additionally, we intend to refine our isosurface evaluation criteria to reduce the number of SDF queries required, improving computational efficiency while maintaining extraction quality.

{
    \bibliographystyle{unsrt}
    \bibliography{main}
}



\vfill

\end{document}